\documentclass{article}
\usepackage{fullpage}
\usepackage{amsmath,amsthm}

\DeclareMathOperator{\tr}{Tr}

\newtheorem{lemma}{Lemma}
\newtheorem{prop}{Proposition}
\newtheorem{theorem}{Theorem}

\begin{document}

\newcommand{\ket}[1]{|{#1}\rangle}
\newcommand{\bra}[1]{\langle{#1}|}
\newcommand{\Hil}{\mathcal{H}}

\title{Entanglement vs.\ gap for one-dimensional spin systems}
\author{Daniel Gottesman$^1$\thanks{E-mail: dgottesman@perimeterinstitute.ca}, and M.~B.~Hastings$^{2,3}$\thanks{E-mail: xhastings@gmail.com}\\
\textit{$^1$Perimeter Institute for Theoretical Physics, Waterloo, Ontario, Canada}
\\
\textit{$^2$Center for Nonlinear Studies and Theoretical Division,}
\\
\textit{Los Alamos National Laboratory, Los Alamos, NM,
USA}
\\
\textit{$^3$ Microsoft Research, Station Q, University of California, Santa Barbara, CA, 93106}
}

\date{}

\maketitle

\abstract{We study the relationship between entanglement and spectral
gap for local Hamiltonians in one dimension.
The area law for a one-dimensional system states that for the ground state, the entanglement of any interval is upper-bounded by a constant independent of the size of the interval.  However, the possible dependence of the upper bound on the spectral gap $\Delta$ is not known,
as the best known general upper bound is asymptotically much
larger than the largest possible entropy of any
model system previously constructed for small $\Delta$.
To help resolve this asymptotic behavior, we construct a family of one-dimensional local systems for which some intervals have entanglement entropy which is polynomial in $1/\Delta$, whereas previously studied systems, such
as free fermion systems or systems described by conformal field theory,
had the entropy of all intervals bounded by a constant times $\log(1/\Delta)$.}

\section{Introduction}

In many local quantum systems, the long-distance behavior of correlation
functions
is a good diagnostic of the critical properties of a system near a phase transition
as the excitation gap $\Delta$ tends to zero.
Entanglement provides a different diagnostic, which
gives complementary information about critical behavior, and is important for our ability
to simulate these systems on classical computers, as described later.

To quantify entanglement,
we define the entanglement entropy of a region $A$ with the complementary region $B$ (assuming the global state on $A \cup B$ is pure) as the von Neumann entropy $S(\rho_A)$, with $\rho_A$ the density matrix of region $A$ (i.e., the global state traced over region $B$).
A state of a spin system that obeys an ``area law'' has the property that the entropy of a region (in one or
more dimensions) is bounded by the surface area of the region instead of the usual maximum of the volume of the
region.  (See~\cite{arealaw} for a review of area laws and references to a full cross-section of work in that area.)
Area laws are conjectured to hold for the ground state of any gapped Hamiltonian spin system with
short-range interactions in finite dimensions, but have not been proved in general.
For one-dimensional systems, the area law simply states that the
entropy of an interval along the line is bounded by a constant,
no matter how large the interval.

At a critical point, it is possible that the entanglement, like the correlation length, can diverge.
However, if we are merely \emph{near} a critical point (either because of additional Hamiltonian
terms or because we are considering a finite-size system), the gap $\Delta$ becomes finite, and
in one dimension, an area law must hold~\cite{Hastings-area-law}.   The entropy can thus provide
another way to identify a critical point.  Note that
diverging entanglement and diverging correlation length do not
necessarily occur together;
on the one hand, there exist states with all correlation
functions short-range but with arbitrarily large entanglement~\cite{expanders,expanders2}, while on the other hand there are
critical points in two dimensions where the correlation functions become long-range but the
area law is still obeyed~\cite{2dboson}.

In one dimension, the entanglement entropy has important implications for our ability to simulate quantum systems
on a classical computer.  Intuitively, so long as the entanglement entropy remains small, the system is ``classical" in some
sense.  More precisely, we know that whenever the system
obeys an area law for the R\'{e}nyi entropy, $S_{\alpha}\equiv (1-\alpha)^{-1}\ln({\rm tr}(\rho_A^{\alpha}))$,
for some $\alpha<1$ (which implies an area law for the von Neumann entropy), there exists a matrix product state approximation to the ground state
with an error inverse polynomial in the bond dimension~\cite{renyiapprox}.  This approximation lies behind the powerful density-matrix
renormalization group algorithm~\cite{dmrg}, as well as behind Vidal's algorithm for time evolution~\cite{tebd}.

Hastings~\cite{Hastings-area-law} proved that a general one-dimensional gapped Hamiltonian spin system satisfies an area law.  However, the bound is probably not tight:
\begin{equation}
S({\rm region}) \leq S_{\rm max} = c_0 \xi' \ln \xi' \ln D \, 2^{\xi' \ln D},
\label{eq:oneDbound}
\end{equation}
where $c_0$ is a constant of order $1$, $D$ is the Hilbert space dimension of a single spin in the system, and
\begin{equation}
\xi' = 6 \max (2v/\Delta, \xi_c),
\end{equation}
where $v$ and $\xi_c$ are parameters from the Lieb-Robinson bound~\cite{LRbound,locality,lr2,lr3}, and $\Delta$ is the spectral gap of the Hamiltonian.
The quantity $v$ is the Lieb-Robinson velocity, and is roughly equal to the size of the energy scale of the largest individual terms in the Hamiltonian times the interaction range.
The Lieb-Robinson velocity sets an upper bound for the group velocity of excitations, so that for a local operator $O$, the
operator $O(t)=\exp(iHt)O\exp(-iHt)$ can be approximately described by an operator supported on the set of sites within
distance $vt$ of the support of $O$.
The quantity $\xi_c$ is a length of order the interaction range which
determines the distance which excitations can ``leak" outside the light cone of the Lieb-Robinson bound.
In particular, all of these constants can be taken to be order $1$, except for $\Delta$, which could be very small even with all
other parameters constant in size; for example, by tuning a system through a quantum critical point the quantity $\Delta$ can be
made small.
Thus, in the bound~(\ref{eq:oneDbound}), the entropy is limited by $\exp [O(1/\Delta)]$.

In contrast, physical intuition suggests that entanglement should be able to reach a distance roughly equal to the correlation length, i.e., about $1/\Delta$.  However, most one-dimensional systems considered by physicists in fact have an even smaller amount of entanglement.  Near a conformal critical point, the bound in the area law diverges as only $\log(1/\Delta)$~\cite{cftent,cftent2}.  Indeed, as far as we aware, for \emph{all} one-dimensional systems where the area law has previously been calculated (including those near non-conformal critical points), the entropy of an interval is also bounded by about $\log (1/\Delta)$, which is below the bound in eq.~(\ref{eq:oneDbound}) by not one but {\em two} exponentials!

Therefore, we raise the following as an open question: given a one-dimensional system with
interaction range and interaction strength both $O(1)$, how big can the entanglement be for a given spectral
gap?  Someone familiar with the previous work on the subject might well conjecture a bound of $\log(1/\Delta)$.
In this paper we falsify that conjecture and get closer to the naive intuitive bound of $1/\Delta$ by demonstrating a family of systems where the entropy of some regions is ${\rm poly} (1/\Delta)$.  In particular, we construct systems on a pair of rings coupled along a junction, each with $N$ $3$-state spins.  In the ground state, the two rings are highly entangled with each other, and the entropy for one ring is roughly $N$.  The spectral gap $\Delta$ is roughly $1/N^4 \log N$, so these systems have entropy for some regions equal to $(-\Delta \log \Delta)^{-1/4}$.  The Hamiltonians we
present depend on $N$ in two ways: First, they are associated to the ring structure, and the size of the rings changes with $N$, and second, we have some terms in the Hamiltonian which are polynomially small in $N$.  Since there are also Hamiltonian terms of size $O(1)$, the presence of some small terms in the Hamiltonian does not affect the bound~(\ref{eq:oneDbound}).  In sec.~\ref{sec:infinite}, we discuss how to let the number of the particles in the system be much larger than $N$, so we can consider $N$ to be a separate parameter from the size of the system.

We have not succeeded in
obtaining the naive estimate of $1/\Delta$ and we are very far from $\exp[O(1/\Delta)]$, so at present neither of the
authors has any firm intuition for a tight upper bound.  Given progress in recent years in understanding the
relation between correlations and spectral gap~\cite{locality}, hopefully this question will be answered too.
Intuitively, an area law is fairly simple: in a gapped system, the entanglement spreads only
over a finite length and therefore only the spins near the boundary of $A$  are
entangled with those outside of $A$.  However, in general any spin within $A$ will be entangled with
other spins within $A$ as well as spins in the complement of $A$.  Thus, to turn this intuition into practice
seems quite tricky as it leads us into the difficult question of multi-party entanglement.

The intuition for the construction used in this paper
is somewhat similar to that used for proofs of QMA-completeness and universality of adiabatic quantum computation~\cite{QMA1D}, although the details are rather different.  In particular, we imagine a computational process that generates entanglement among the spins and write down a Hamiltonian whose ground state is a superposition of time steps of the process.

To be more precise, we imagine the spins have the three states $\ket{0}$, $\ket{1}$, and $\ket{x}$.  $\ket{0}$ and $\ket{1}$ are two states of a qubit, and most of the spins in each ring will be in this qubit subspace, but one spin in each ring will be in the $\ket{x}$ state.  We can imagine the $\ket{x}$ as a ``hole'' in a ring otherwise filled with spin-$1/2$ particles.  The rings will meet in two sites, which we number $1$ and $N$ on each ring, with sites $2, \ldots, N-1$ on each ring being connected only to points in order on the same ring.  When there are two qubit states at the two number-$1$ sites, we project onto a maximally entangled state such as the singlet $\ket{0}_L \ket{1}_R - \ket{1}_L \ket{0}_R$.  We then move each $\ket{x}$ state around its ring, shifting the qubit states by one place, allowing us to entangle another pair.  The steady-state behavior of this system is thus $N-1$ maximally entangled pairs, aligned roughly between corresponding sites on the two rings, perhaps offset by one place, depending on the locations of the $\ket{x}$ states.

We then write down a nearest-neighbor Hamiltonian whose ground state is the superposition of such states over all possible locations of the two $\ket{x}$ states.  In the ground state, the entropy of one ring taken alone will be slightly over $N-1$, since we broken all $N-1$ EPR pairs, but we will show that the spectral gap of the Hamiltonian is $\Delta = \Omega(1/N^4 \log N)$.  Thus, the entropy of just the left ring, say, is $\Omega((-\Delta \log \Delta)^{-1/4})$, a polynomial in $1/\Delta$.

The ground state differs from those used in the previous QMA-completeness constructions in that there is an $\ket{x}$ state in each ring, and we take the superposition over possible positions of the {\em pair}, whereas in the previous constructions, the ground state was a superposition over time slices.  Essentially, the previous Hamiltonians were quantum walks on a line, whereas ours is a walk on a two-dimensional graph.

While the Hamiltonian we write requires fine-tuning, it is not too unnatural from the physics point of view.  Regarding
the state $\ket{0}$ as representing an electron with spin up, the state $\ket{1}$ as an electron with spin down, and
$\ket{x}$ as a hole, within the rings our Hamiltonian is simply the $t-J$ model Hamiltonian of condensed matter physics with $t\neq 0$ and $J=0$.  Only
where the rings meet do we have a more complicated interaction.

Sandy Irani has shown a similar result independently~\cite{irani}.
The two constructions have slightly different motivations and each
has different advantages.
We were primarily interested in the dependence of entanglement
entropy on gap, and we
have a more rapid divergence of entanglement entropy with inverse
gap and use fewer states per site.  Irani's primary motivation was to
study the effect of
translational invariance on entanglement entropy.  Her construction
shows that it is possible to achieve a polynomial dependence of
entanglement entropy on gap with a translation-invariant Hamiltonian, whereas our construction explicitly breaks
translation invariance.

\section{From Two Rings To An Infinite Line}
\label{sec:infinite}

As mentioned above, our main concern in this paper is to construct a local Hamiltonian $H$ for a system of $2N$ $3$-state spins.  It will consist of two rings containing $N$ sites each which are connected at two adjacent points, and we will show that for the ground state of $H$, the entanglement of the left ring with the right ring is $\Omega((-\Delta \log \Delta)^{-1/4})$ (where $\Delta$ is the spectral gap of $H$).  It may not be immediately obvious how this relates to the area law of an infinite one-dimensional system, so before we discuss the technical details of the construction, let us first explain how the result immediately carries over to a system of spins on a line with the property that the entanglement of some regions is polynomial in the inverse gap.

First, note that we can squash the two rings down into a one-dimensional ``brick'' composed of a total of $N$ $9$-state spins, as illustrated in figure~\ref{fig:brick}.
\begin{figure}
\begin{center}
\begin{picture}(300,60)

\put(10,45){\makebox(20,10){a)}}
\put(50,30){\circle{40}}
\put(100,30){\circle{40}}

\put(69,36){\circle*{5}}
\put(62,46){\circle*{5}}
\put(50,50){\circle*{5}}
\put(38,46){\circle*{5}}
\put(31,36){\circle*{5}}
\put(31,24){\circle*{5}}
\put(38,14){\circle*{5}}
\put(50,10){\circle*{5}}
\put(62,14){\circle*{5}}
\put(69,24){\circle*{5}}

\put(62,34){\makebox(0,0){1}}
\put(55,43){\makebox(0,0){2}}
\put(45,40){\makebox(0,0){\ldots}}
\put(62,24){\makebox(0,0){N}}

\put(50,30){\oval(12,12)[tr]}
\put(50,30){\oval(12,12)[tl]}
\put(50,30){\oval(12,12)[bl]}
\put(50,24){\vector(1,0){5}}

\put(119,36){\circle*{5}}
\put(112,46){\circle*{5}}
\put(100,50){\circle*{5}}
\put(88,46){\circle*{5}}
\put(81,36){\circle*{5}}
\put(81,24){\circle*{5}}
\put(88,14){\circle*{5}}
\put(100,10){\circle*{5}}
\put(112,14){\circle*{5}}
\put(119,24){\circle*{5}}

\put(88,34){\makebox(0,0){1}}
\put(95,43){\makebox(0,0){2}}
\put(105,40){\makebox(0,0){\ldots}}
\put(88,24){\makebox(0,0){N}}

\put(100,30){\oval(12,12)[tl]}
\put(100,30){\oval(12,12)[tr]}
\put(100,30){\oval(12,12)[br]}
\put(100,24){\vector(-1,0){5}}

\put(160,45){\makebox(20,10){b)}}
\put(180,24){\framebox(48,12){}}
\put(238,24){\framebox(48,12){}}

\put(228,36){\circle*{5}}
\put(216,36){\circle*{5}}
\put(204,36){\circle*{5}}
\put(192,36){\circle*{5}}
\put(180,36){\circle*{5}}
\put(180,24){\circle*{5}}
\put(192,24){\circle*{5}}
\put(204,24){\circle*{5}}
\put(216,24){\circle*{5}}
\put(228,24){\circle*{5}}

\put(228,46){\makebox(0,0){1}}
\put(216,46){\makebox(0,0){2}}
\put(204,43){\makebox(0,0){\ldots}}
\put(228,14){\makebox(0,0){N}}

\put(286,36){\circle*{5}}
\put(274,36){\circle*{5}}
\put(262,36){\circle*{5}}
\put(250,36){\circle*{5}}
\put(238,36){\circle*{5}}
\put(238,24){\circle*{5}}
\put(250,24){\circle*{5}}
\put(262,24){\circle*{5}}
\put(274,24){\circle*{5}}
\put(286,24){\circle*{5}}

\put(238,46){\makebox(0,0){1}}
\put(250,46){\makebox(0,0){2}}
\put(262,43){\makebox(0,0){\ldots}}
\put(238,14){\makebox(0,0){N}}

\put(174,18){\dashbox{3}(12,24){}}

\end{picture}
\end{center}
\caption{a) The two rings.  Each black dot represents a $3$-state spin.  b) The squashed rings.  Now, vertical pairs of black dots represent a single $9$-state spin.}
\label{fig:brick}
\end{figure}
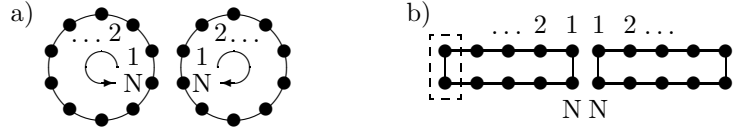
Note that $3$-state spins which are neighboring in the rings get mapped to neighboring $9$-state spins in the brick, so the resulting Hamiltonian is still a nearest-neighbor Hamiltonian.  While at the place where the two rings meet, the ring Hamiltonian has a four-spin term (see the following section), all four of those spins are contained in just two $9$-state spins, so the brick Hamiltonian contains only two-spin terms.  Note that it is important that we squash the rings this way and not try to connect the rings along their length, as we are interested in studying the area law based on the entropy of a single ring by itself, and we therefore want the surface area of a single ring to be small.  In the brick, the surface area of what used to be a ring is simply $2$ spins: the $9$-state spin consisting of the original sites $1$ and $N$, and the $9$-state spin consisting of sites $N/2$ and $N/2+1$.  If we had joined the rings length-wise, the surface area would be $N$, the size of a ring, which would be no good.

Now we have a finite truly one-dimensional system whose ground state has the property that the left half of the brick has high entanglement with the right half.  We can take an infinite line of $9$-state particles and chop it up into identical bricks of this sort (figure~\ref{fig:line}).
\begin{figure}
\begin{center}
\begin{picture}(400,60)

\multiput(30,24)(58,0){6}{\framebox(48,12){}}

\multiput(78,36)(58,0){6}{\circle*{5}}
\multiput(66,36)(58,0){6}{\circle*{5}}
\multiput(54,36)(58,0){6}{\circle*{5}}
\multiput(42,36)(58,0){6}{\circle*{5}}
\multiput(30,36)(58,0){6}{\circle*{5}}
\multiput(30,24)(58,0){6}{\circle*{5}}
\multiput(42,24)(58,0){6}{\circle*{5}}
\multiput(54,24)(58,0){6}{\circle*{5}}
\multiput(66,24)(58,0){6}{\circle*{5}}
\multiput(78,24)(58,0){6}{\circle*{5}}

\multiput(26,20)(116,0){3}{\framebox(114,20){}}

\put(6,24){\makebox(20,12){$\cdots$}}
\put(372,24){\makebox(20,12){$\cdots$}}

\end{picture}
\end{center}
\caption{Part of an infinite system of bricks.}
\label{fig:line}
\end{figure}
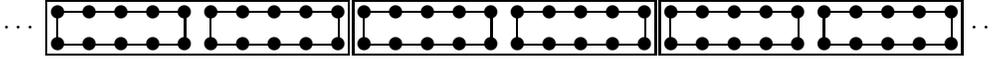
If we take a Hamiltonian which is simply the sum of the brick Hamiltonians on each piece, with no terms interacting distinct pieces, the ground state of the resulting Hamiltonian is simply the tensor product of the ground states of the individual bricks.  There are many excited states, but the lowest-lying excited states are superpositions of states for which all but one brick is in the ground state, and the remaining brick is in its first excited state.  The gap of the line Hamiltonian is thus just equal to $\Delta$, the gap of the original ring Hamiltonian (i.e., $\Omega(1/N^4 \log N)$).  Let us examine an interval $A$ which begins at the center of one brick and ends at the center of another brick.  Bricks which are completely contained in $A$ do not contribute to the entropy, but the two bricks which are cut in half each contribute about $N$ units of entropy, so the entropy of $A$ in the ground state is roughly $2N$.  In particular, the entropy of $A$ is $\Omega((\Delta \log \Delta)^{-1/4})$.

Of course, some regions have very low entropy.  Indeed, an interval containing only complete bricks will have entropy $0$.  This does not affect the area law for the system, which simply provides an upper bound on the entanglement of any region.  Some regions have high entropy whereas others do not. Nevertheless, if one finds this result unsatisfying, one may always consider a larger system constructed by two adjacent lines of $9$-state particles, each constructed as above, but offset by half a brick (as in figure~\ref{fig:overlapping}).
\begin{figure}
\begin{center}
\begin{picture}(400,60)

\multiput(30,10)(58,0){6}{\framebox(48,12){}}

\multiput(78,22)(58,0){6}{\circle*{5}}
\multiput(66,22)(58,0){6}{\circle*{5}}
\multiput(54,22)(58,0){6}{\circle*{5}}
\multiput(42,22)(58,0){6}{\circle*{5}}
\multiput(30,22)(58,0){6}{\circle*{5}}
\multiput(30,10)(58,0){6}{\circle*{5}}
\multiput(42,10)(58,0){6}{\circle*{5}}
\multiput(54,10)(58,0){6}{\circle*{5}}
\multiput(66,10)(58,0){6}{\circle*{5}}
\multiput(78,10)(58,0){6}{\circle*{5}}

\multiput(26,6)(116,0){3}{\framebox(114,20){}}

\put(6,10){\makebox(20,12){$\cdots$}}
\put(372,10){\makebox(20,12){$\cdots$}}

\multiput(30,35)(58,0){6}{\framebox(48,12){}}

\multiput(78,47)(58,0){6}{\circle*{5}}
\multiput(66,47)(58,0){6}{\circle*{5}}
\multiput(54,47)(58,0){6}{\circle*{5}}
\multiput(42,47)(58,0){6}{\circle*{5}}
\multiput(30,47)(58,0){6}{\circle*{5}}
\multiput(30,35)(58,0){6}{\circle*{5}}
\multiput(42,35)(58,0){6}{\circle*{5}}
\multiput(54,35)(58,0){6}{\circle*{5}}
\multiput(66,35)(58,0){6}{\circle*{5}}
\multiput(78,35)(58,0){6}{\circle*{5}}

\multiput(84,31)(116,0){2}{\framebox(114,20){}}
\put(21,31){\line(1,0){61}}
\put(82,31){\line(0,1){20}}
\put(82,51){\line(-1,0){61}}
\put(377,31){\line(-1,0){61}}
\put(316,31){\line(0,1){20}}
\put(316,51){\line(1,0){61}}

\put(6,35){\makebox(20,12){$\cdots$}}
\put(372,35){\makebox(20,12){$\cdots$}}

\end{picture}
\end{center}
\caption{Two layers of offset bricks.}
\label{fig:overlapping}
\end{figure}
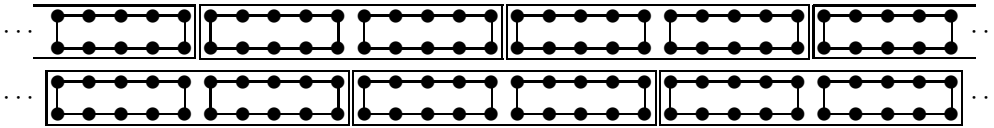
Now any interval will cut off on each end at least one-quarter of a brick on one of the two lines.  We can combine the two lines into a single line of $81$-state particles, and the resulting Hamiltonian has the property that for its ground state, any interval of at least $N$ spins will have entropy $\Omega((- \Delta \log \Delta)^{-1/4})$.

\section{The Two-Ring Hamiltonian}

We will write our Hamiltonian $H$ as a sum of various terms:
\begin{equation}
H = H_L + H_R + H_B + H_P + H_V.
\end{equation}
The terms $H_L$ and $H_R$ cause the spins (containing a qubit state) to hop around the left and right rings, respectively.  $H_B$ handles the hopping of the $\ket{x}$ ``holes'' between the sites numbered $1$ and $N$ on each ring.  In order to keep the two rings in sync, we require that the holes hop in unison when they hop between sites $1$ and $N$, and $H_B$ ensures that this happens.  $H_P$ is the projection Hamiltonian, which is responsible for ensuring that any pair of qubits at site $1$ on the left and right rings are in fact in the singlet state.  $H_x$ is a self-energy term for the $\ket{x}$ hole states, ensuring that the ground state has exactly one $\ket{x}$ in each ring.

More specifically, let
\begin{equation}
S_i^{L/R} = \sum_{a,b} \ket{ab}_{i,i+1}^{L/R} \bra{ba}.
\end{equation}
We work modulo $N$, so site $N+1$ should be considered the same as site $1$.  $S_i^{L/R}$ swaps the particles in sites $i$ and $i+1$ for the left (L) or right (R) ring, regardless of the values of the spins at those sites.  Thus we can define
\begin{equation}
F_i^{L/R} = S_i^{L/R} \ket{x}_i^{L/R} \bra{x},
\end{equation}
which swaps an $\ket{x}$ at site $i$ with a state at site $i+1$, and annihilates any state which does not have an $\ket{x}$ at site $i$.

We then let
\begin{equation}
H_{L/R} = \sum_{i=1}^{N-1} \left( \ket{x}_i^{L/R} \bra{x} + \ket{x}_{i+1}^{L/R} \bra{x} - F_i^{L/R} - (F_i^{L/R})^\dagger \right).
\end{equation}
$H_L$ thus produces an energy penalty for having the $\ket{x}$ in the $i$th site in the left ring, but cancels that energy penalty if the $\ket{x}$ hops to the $(i+1)$th site, swapping with the state there; similarly, there is a penalty for having the $\ket{x}$ in the $(i+1)$th site of the left ring which is cancelled by letting it hop to the $i$th site.  $H_L$ thus tends to produce a superposition of $\ket{x}$ over sites $1$ through $N$ in the left ring, with the qubit states shifting appropriately to make room, and $H_R$ favors a similar superposition over sites in the right ring.  Note, however, that $H_L$ and $H_R$ do not have terms allowing $\ket{x}$ to hop between sites $1$ and $N$.

Instead, we let $H_B$ provide those terms:
\begin{equation}
H_B = \ket{xx}_N^{LR} \bra{xx} + \ket{xx}_1^{LR} \bra{xx} - F_N^L F_N^R - (F_N^L)^\dagger (F_N^R)^\dagger.
\end{equation}
$H_B$ has a structure similar to $H_L$ and $H_R$, but it instead provides a penalty only if the $\ket{x}$ states for the two rings are both at site $1$ or both at site $N$, and then gives a cancelling energy bonus if they hop together to the other pair (i.e., both hop from site $N$ to site $1$ or from site $1$ to site $N$).  In conjunction with $H_L$ and $H_R$, this means that the $\ket{x}$ states in the left and right rings can wander separately in the two rings, but cannot complete a full circle around a ring except by hopping together.

$H_P$ is defined as a projector:
\begin{equation}
H_P = (I - \ket{x}_1^L \bra{x}) (I - \ket{x}_1^R \bra{x}) - \ket{\Psi^-}_1^{LR} \bra{\Psi^-},
\end{equation}
which produces an energy penalty for any state which has a qubit state in site number $1$ for both rings, unless the two qubit states form a singlet state $\ket{\Psi^-} = (\ket{01}-\ket{10})/\sqrt{2}$.

Finally, let
\begin{equation}
H_V^{L/R} = V_1-V_1 \sum_{i=1}^N \ket{x}_i^{L/R} \bra{x} + V_2 \sum_{i=1}^N \ket{xx}_{i,i+1}^{L/R} \bra{xx},
\end{equation}
and
\begin{equation}
H_V = H_V^L + H_V^R.
\end{equation}
We choose $V_1$ and $V_2$ to be positive constants, so $H_V$ will produce an energy bonus proportional to the total number of $\ket{x}$ states in the two rings, but an energy penalty if any two $\ket{x}$ states are adjacent in a single ring.  Since $H_L$, $H_R$, and $H_B$ will tend to produce superpositions of different locations for the $\ket{x}$ states in each ring, if there is more than one $\ket{x}$ within a given ring, $H_L$, $H_R$, and $H_B$ will favor states where some $\ket{x}$ states are adjacent, which will then get an energy penalty due to the $V_2$ term.  By tuning $V_1$ and $V_2$ appropriately, we will ensure that the ground state has exactly one $\ket{x}$ state in each ring, as desired.  The constant term $2V_1$ then gives that ground state $0$ energy.

Note that the terms $H_L$, $H_R$, $H_B$, and $H_P$ are all non-negative operators.  However, $H_V^L$ and $H_V^R$ are not, which means we will need to exercise some caution in bounding the energy gap.

\section{The Ground State and the Gap}

We now turn to determining the structure of the ground state of $H$ and its spectral gap.  We shall prove the following:
\begin{theorem}
Suppose $V_1 = 1/N^4$ and $V_2 = 1$.  Then
\begin{enumerate}
\item[a)] $H$ has a nondegenerate ground state $\ket{g}$,
\item[b)] The spectral gap $\Delta$ of $H$ is $\Omega(1/N^4 \log N)$, and
\item[c)] $S(\rho_L) = S(\rho_R) \geq N-1$, where $\rho_{L/R} = \tr_{R/L} \ket{g} \bra{g}$.
\end{enumerate}
\end{theorem}

The first thing to notice is that no term of $H$ changes the number of $\ket{x}$ states in a ring.  Therefore, we can treat the Hilbert space as a direct sum of subspaces with different numbers of $\ket{x}$ states in the two rings and analyze them separately.  Let the subspace with $a$ $\ket{x}$ states in the left ring and $b$ $\ket{x}$ states in the right ring be $\Hil_{a,b}$. We start with $\Hil_{1,1}$, which we claim contains the global ground state.  We therefore analyze the ground state (which will have energy $0$) and gap within this subspace.  Then we show that all states in the other subspaces have energy at least $O(1/N^4)$.

\subsection{One $\ket{x}$ in Each Ring}

Within the subspace $\Hil_{1,1}$, the Hamiltonian term $H_V$ becomes uniformly $0$, so we can ignore it.  We are thus left with $H = H_L + H_R + H_B + H_P$, and we wish to find the ground state and gap of this Hamiltonian on $\Hil_{1,1}$.

First, let us define some convenient basis states.  We start with states which have both $\ket{x}$ states in site $N$ in their respective rings, and define
\begin{equation}
\ket{\{\alpha_i\}} = \ket{\alpha_1}_1^{LR} \ket{\alpha_2}_2^{LR} \cdots \ket{\alpha_{N-1}}_{N-1}^{LR} \ket{xx}_N^{LR},
\end{equation}
where $\ket{\alpha_i}$ is a Bell state, $i = 1, \ldots, N-1$.  The state $\ket{\{\alpha_i\}}$ thus has the two qubit states at the $j$th position in each ring sharing the particular entangled state $\ket{\alpha_j}$.

Now define $M_{a,b}$ to be an operator that shifts $\ket{x}$ to site $a$ in the left ring and site $b$ in the right ring:
\begin{equation}
M_{a,b} = \left( S_a^{L} S_{a+1}^{L} \cdots S_{N-1}^{L} \right) \left( S_b^{R} S_{b+1}^{R} \cdots S_{N-1}^{R} \right).
\end{equation}
(Make the convention that if $a$ or $b$ is equal to $N$, do not include the L or R operators, respectively.  Thus $M_{N,N} = I$.)
Then $M_{a,b} \ket{\{\alpha_i\}}$ has the $\ket{x}$ states at sites $a$ and $b$ in the left and right rings, respectively, but some of the entangled states $\ket{\alpha_i}$ are now misaligned: If $j < \min\{a,b\}$ then the two sites $j$ share $\ket{\alpha_j}$, and if $j > \max\{a,b\}$, then the two sites $j+1$ share $\ket{\alpha_j}$.  However, if $a < j < b$, then the state $\ket{\alpha_j}$ is shared between site $j+1$ in the left ring and site $j$ in the right ring, whereas if $b < j < a$, then $\ket{\alpha_j}$ is shared between site $j$ in the left ring and $j+1$ in the right ring.  Note that for $M_{N,1} \ket{\{\alpha_i\}}$, all pairs are misaligned by one.  Thus, even though $S_N^R M_{N,1} \ket{\{\alpha_i\}}$ has both $\ket{x}$ sites back at site $N$, the entangled pairs are misaligned, so it is not a state of the form $\ket{\{\beta_i\}}$ for any $\{\beta_i\}$.

Let $R = M_{1,1}^\dagger S_N^L S_N^R$.  Then $R\, \ket{\{\alpha_i\}}$ also has both $\ket{x}$ states back at site $N$, but since the two $\ket{x}$ states have {\em both} circled the ring once, the pairs remain in sync.  They have all been shifted by one place:
\begin{equation}
R\, \ket{\{\alpha_i\}} = \ket{\{\beta_i\}},
\end{equation}
with $\ket{\beta_i} = \ket{\alpha_{i+1}}$ ($i = 1, \ldots, N-2$) and $\ket{\beta_{N-1}} = \ket{\alpha_{1}}$.

We now look at the action of the Hamiltonian on these states.  The states $M_{a,b} \ket{\{\alpha_i\}}$ form a basis for $\Hil_{1,1}$, and $H_P$ is diagonal in this basis, with eigenvalue $0$ for any state of the form $M_{1,b} \ket{\{\alpha_i\}}$, $M_{a,1} \ket{\{\alpha_i\}}$ (for any $a$, $b$), or for $M_{a,b} \ket{\{\alpha_i\}}$, with $a,b > 1$ and $\ket{\alpha_1} = \ket{\Psi^-}$.  $H_L$ and $H_R$ will interact different states $M_{a,b} \ket{\{\alpha_i\}}$ with the same list $\{\alpha_i\}$ of Bell states.  $H_B$ will couple $M_{1,1} \ket{\{\alpha_i\}}$ with $R\, \ket{\{\alpha_i\}}$, so the combination $H_L + H_R + H_B$ couples all states of the form $M_{a,b} \ket{\{\alpha_i\}}$ with the list of Bell states $\{\alpha_i\}$ fixed up to cyclic shifts.  Suppose the list $\{\alpha_i\}$ has period $p$ under cyclic shifts (i.e., $R^p \ket{\{\alpha_i\}} = \ket{\{\alpha_i\}}$).  Then it follows that
\begin{equation}
\ket{S(\{\alpha_i\})} \equiv \frac{1}{\sqrt{pN^2}} \sum_{r=0}^{p-1} \sum_{a,b} M_{a,b} R^r \ket{\{\alpha_i\}}
\end{equation}
is an eigenvector of $H_L + H_R + H_B$ with eigenvalue $0$, and all states annihilated by $H_L + H_R + H_B$ are of this form.  We note that the Hilbert subspace $\Hil_{1,1}$ further decomposes into invariant subspaces of $H$, each of which can be labeled by a list of Bell states $\{\alpha_i\}$, with two lists equivalent if they differ only by a cyclic shift on $i$.  We will refer to the subspace associated to the list $\{\alpha_i\}$ as $\Hil(\{\alpha_i\})$; we choose some convenient representative list $\{\alpha_i\}$ from the set of equivalent lists to label each $\Hil(\{\alpha_i\})$.

The state $\ket{g} = \ket{S(\{\alpha_i\})}$ with $\ket{\alpha_i} = \ket{\Psi^-}$ for all $i$ also has eigenvalue $0$ for $H_P$, as every term in it has either $\ket{x}$ in site $1$ for at least one ring or has a singlet shared between the two $1$ sites.  All other states $\ket{S(\{\alpha_i\})}$ do not have this property: Even if $\ket{\alpha_1} = \ket{\Psi^-}$, some $\ket{\alpha_i} \neq \ket{\Psi^-}$, and therefore $H_P R^{i-1} \ket{\{\alpha_i\}} = 1$.  $\ket{g}$ is thus the ground state within $\Hil_{1,1}$ (and, we will show, within the whole Hilbert space), whereas we claim any other state has substantially higher energy.

\begin{theorem}
\label{thm:onexgap}
If $\ket{\psi} \in \Hil(\{\alpha_i\})$ with $\ket{\alpha_i} \neq \ket{\Psi^-}$ for at least one value of $i$, then
\begin{equation}
\bra{\psi} H_L + H_R + H_B + H_P \ket{\psi} = \Omega(1/p^2 N^2 \log N),
\end{equation}
where $p$ is the periodicity of $\{\alpha_i\}$ under cyclic permutations.
\end{theorem}

Since $p \leq N$, the theorem implies that the gap of the full Hamiltonian within $\Hil_{1,1}$ is $\Omega(1/N^4 \log N)$.

As a warmup to proving this, we first consider $H_L + H_R + H_B$ and prove that it has a polynomial gap.  While we will not use precisely this proposition in the proof of Theorem~\ref{thm:onexgap}, the proof uses a very similar technique, and the main lemma for Prop.~\ref{prop:xmovegap} will also be central to proving Theorem~\ref{thm:onexgap}.

The Hamiltonian $H_L + H_R + H_B$ can be thought of as a quantum walk on a graph.  In particular, let us restrict to the invariant subspace $\Hil(\{\alpha_i\})$, where $\{\alpha_i\}$ has period $p$ under cyclic shifts.  We then get a graph $G_p$ of $pN^2$ nodes.  The states $\ket{a,b,r} \equiv M_{a,b} R^r \ket{\{\alpha_i\}}$ can be considered as the nodes of the graph, which consists of $p$ $N \times N$ grids, with the $(N,N)$ site of one grid connected to the $(1,1)$ site of the next one, and the $(N,N)$ site of the $p$th grid connected to the $(1,1)$ site of the $1$st grid.  (See figure~\ref{fig:grids}.)
%
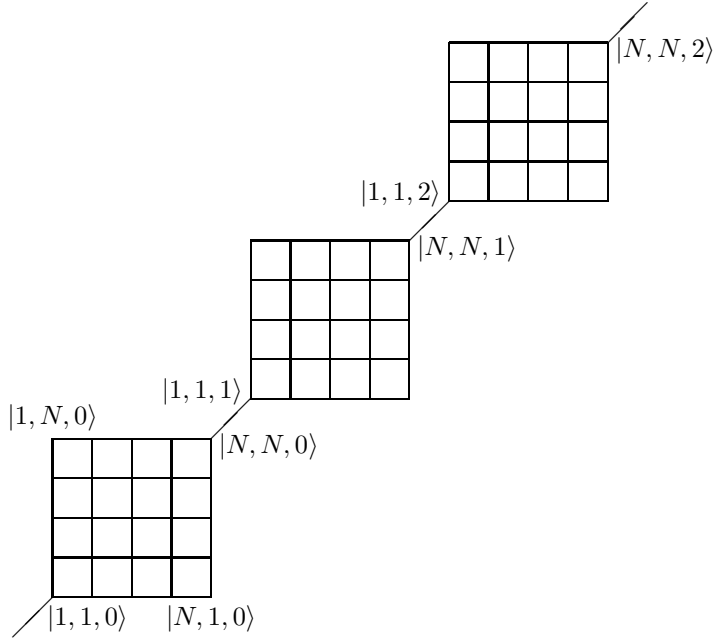
\begin{figure}
\begin{center}
\begin{picture}(240,240)

\put(0,0){\line(1,1){15}}

\multiput(15,15)(0,15){5}{\line(1,0){60}}
\multiput(15,15)(15,0){5}{\line(0,1){60}}

\put(13,12){\makebox(0,0)[tl]{$\ket{1,1,0}$}}
\put(75,12){\makebox(0,0)[t]{$\ket{N,1,0}$}}
\put(15,78){\makebox(0,0)[b]{$\ket{1,N,0}$}}
\put(78,77){\makebox(0,0)[tl]{$\ket{N,N,0}$}}

\put(75,75){\line(1,1){15}}

\multiput(90,90)(0,15){5}{\line(1,0){60}}
\multiput(90,90)(15,0){5}{\line(0,1){60}}

\put(87,88){\makebox(0,0)[br]{$\ket{1,1,1}$}}
\put(153,152){\makebox(0,0)[tl]{$\ket{N,N,1}$}}

\put(150,150){\line(1,1){15}}

\multiput(165,165)(0,15){5}{\line(1,0){60}}
\multiput(165,165)(15,0){5}{\line(0,1){60}}

\put(162,163){\makebox(0,0)[br]{$\ket{1,1,2}$}}
\put(228,227){\makebox(0,0)[tl]{$\ket{N,N,2}$}}

\put(225,225){\line(1,1){15}}

\end{picture}
\end{center}
\caption{The graph of basis states within $\Hil(\{\alpha_i\})$.}
\label{fig:grids}
\end{figure}
Edges of the graph correspond to states which are coupled by a single off-diagonal term of the Hamiltonian.

\begin{prop}
\label{prop:xmovegap}
The gap of $H_L + H_R + H_B$ within $\Hil(\{\alpha_i\})$ is $\Omega(1/p^2 N^2 \log N)$.
\end{prop}

\begin{proof}[Proof of proposition]
$H_L$ is a simple quantum walk on a line of length $N$.  It can be explicitly diagonalized, with eigenstates
\begin{equation}
\ket{\psi_{k},b,r} = \sqrt{\frac{2}{N}} \sum_{a=1}^N \cos \left(\frac{\pi k}{N} (a-1/2)\right) \ket{a,b,r}
\end{equation}
and eigenvalues
\begin{equation}
\lambda_k = 2(1 - \cos \frac{\pi k}{N})
\end{equation}
for $k = 1, \ldots, N-1$.  $k=0$ is also possible (although with a different normalization constant), and gives the ground state $\ket{\psi_0,b,r} = (1/\sqrt{N}) \sum_a \ket{a,b,r}$, with energy $0$.  Similarly for $H_R$.  Since a single grid has the two coordinates completely uncoupled, the eigenstates $\ket{\psi_{kl},r}$ of $H_L + H_R$ are similarly simple, with eigenvalues $\lambda_{kl} = \lambda_k + \lambda_l$.  The ground state of $H_L + H_R$ (for fixed $r$) is $\ket{\psi_{00},r} = (1/N) \sum_{a,b} \ket{a,b,r}$.  For general $k, l$, let us write $\ket{\psi_{kl},r} = \sum_{a,b} f_{kl}(a,b) \ket{a,b,r}$.  Then
\begin{equation}
f_{kl}(a,b) = C_{kl} \cos\left(\frac{\pi k}{N} (a-1/2)\right) \cos\left(\frac{\pi l}{N} (b-1/2)\right),
\end{equation}
where $C_{kl}$ is a normalization constant: $C_{kl} = 1/N$ for $(k,l) = (0,0)$, $C_{kl} = \sqrt{2}/N$ if exactly one of $k$, $l$ is $0$, and $C_{kl} = 2/N$ otherwise.

The ground state of $H_L + H_R + H_B$ is again simple:
\begin{equation}
\ket{{\rm ground}} = \frac{1}{\sqrt{p}} \sum_{r=0}^{p-1} \ket{\psi_{00},r} = \frac{1}{\sqrt{p} N} \sum_{r=0}^{p-1} \sum_{a,b = 1}^{N} \ket{a,b,r}.
\end{equation}
However, the higher eigenstates are more complicated.  We can simplify slightly by noting that the graph has translation invariance under the action of $R$, so the eigenstates of $H_L + H_R + H_B$ will also be eigenstates of $R$.  $R$ has eigenvalues $\exp(2\pi i m/p)$ (since $R^p = I$ on the subspace $\Hil(\{\alpha_i\})$).  Any eigenstate $\ket{\psi}$ of $R$ can be written as
\begin{equation}
\ket{\psi} = \frac{1}{\sqrt{p}} \sum_{r=0}^{p-1} \sum_{k,l=0}^{N-1} c_{kl} e^{2\pi i mr/p} \ket{\psi_{kl},r},
\label{eigenstate}
\end{equation}
with $m = 0, \ldots, p-1$.  Let us consider the expected energy $\bra{\psi} H_L + H_R + H_B \ket{\psi}$ of any state $\ket{\psi}$ which is orthogonal to the ground state $\ket{{\rm ground}}$ and is an eigenstate of $R$.  If $m=0$, $c_{00} = 0$ (because the state is orthogonal to the ground state), so
\begin{equation}
\bra{\psi} H_L + H_R \ket{\psi} \geq \lambda_1 = \Theta(1/N^2).
\end{equation}

Therefore, we should search for the first excited state among states with nonzero $m$.  In particular, we will show that for any state with $m>0$, either the state is dominated by the ground state of $H_L + H_R$ (i.e., $c_{00}$ is very near $1$), in which case it does not satisfy $H_B$ very well, or there is a significant mixture of higher eigenstates of $H_L + H_R$.

For the proof of Theorem~\ref{thm:onexgap}, we will need a stronger version of the result using unnormalized states.  The necessary result is contained in the following lemma:
\begin{lemma}
\label{lemma:highm}
Let
\begin{equation}
\ket{\psi} = \frac{1}{\sqrt{p}} \sum_{r=0}^{p-1} \sum_{k,l=0}^{N-1} c_{kl} e^{2\pi i mr/p} \ket{\psi_{kl},r}.
\end{equation}
Then
\begin{equation}
\bra{\psi} H_L + H_R + H_B \ket{\psi} = \Omega(m'^2 |c_{00}|^2/p^2 N^2 \log N),
\end{equation}
where $m' = \min(m, p-m)$.
\end{lemma}

When $\ket{\psi}$ is normalized, it is easy to see that $\bra{\psi} H_L + H_R \ket{\psi}$ is large ($\Omega(1/N^2)$) unless $|c_{00}|$ is order $1$, so the proposition follows immediately from the lemma.

\end{proof}

\begin{proof}[Proof of lemma~\ref{lemma:highm}]

Let us consider the expectation value of $H_B$ for $m > 0$:
\begin{equation}
\bra{\psi} H_B \ket{\psi} = \frac{1}{p} \sum_{r,r' = 0}^{p-1} \sum_{k,l,k',l' = 0}^{N-1} c_{kl} c^*_{k'l'} e^{2\pi i m (r-r')/p} \bra{\psi_{k'l'},r'} H_B \ket{\psi_{kl},r}.
\end{equation}
Since
\begin{multline}
\bra{\psi_{k'l'},r'} H_B \ket{\psi_{kl},r} = \delta_{r,r'} \left[ f_{kl}(1,1) f^*_{k'l'}(1,1) + f_{kl}(N,N) f^*_{k'l'}(N,N)  \right] \\
- \delta_{r,r'-1} f_{kl}(N,N) f^*_{k'l'}(1,1) - \delta_{r,r'+1} f_{kl}(1,1) f^*_{k'l'}(N,N),
\end{multline}
we have, focusing on the $(k,l) = (0,0)$ terms,
\begin{multline}
\bra{\psi} H_B \ket{\psi} =
\frac{1}{N^2} |c_{00}|^2 \left[ 2 - 2 \cos (2 \pi i m /p) \right] + \bra{\psi'} H_B \ket{\psi'}
\\ +
\frac{1}{N} \sum_{(k,l) \neq (0,0)} c_{kl} c^*_{00} \left[ f_{kl}(1,1)(1 - e^{2\pi i m/p}) + f_{kl}(N,N) (1 - e^{-2\pi i m/p}) \right]
\\ +
\frac{1}{N} \sum_{(k',l') \neq (0,0)} c_{00} c^*_{k'l'} \left[ f_{k'l'}^*(1,1)(1 - e^{-2\pi i m/p}) + f_{k'l'}^*(N,N) (1 - e^{2\pi i m/p}) \right],
\end{multline}
where
\begin{equation}
\ket{\psi'} = \frac{1}{\sqrt{p}} \sum_{r=0}^{p-1} \sum_{(k,l) \neq (0,0)} c_{kl} e^{2\pi i m r/p} \ket{\psi_{kl},r}.
\end{equation}
Since $H_B$ is positive semidefinite, $\bra{\psi'} H_B \ket{\psi'} \geq 0$, and
\begin{align}
\bra{\psi} H_B \ket{\psi} \geq &
\ \frac{1}{N^2} |c_{00}|^2 \left[ 2 - 2 \cos (2 \pi m /p) \right] \\
& \quad +
2 {\rm Re} \left\{ \frac{1}{N} \sum_{k,l \neq (0,0)} c_{kl} c^*_{00} \left[ f_{kl}(1,1)(1 - e^{2\pi i m/p}) + f_{kl}(N,N) (1 - e^{-2\pi i m/p}) \right] \right\} \nonumber \\
= &
\ \frac{2}{N^2} |c_{00}|^2 \left[ 1 - \cos (2 \pi m /p) \right] \\
& \quad +
\frac{2}{N} \sum_{k,l \neq (0,0)} {\rm Re} (c_{kl} c^*_{00}) C_{kl} \cos(\pi k/2N) \cos(\pi l/2N) [1 - \cos (2\pi m/p)] [1  + (-1)^{k+l}]
\nonumber \\ & \quad +
\frac{2}{N} \sum_{k,l \neq (0,0)} {\rm Im} (c_{kl} c^*_{00}) C_{kl} \cos(\pi k/2N) \cos(\pi l/2N) \sin (2\pi m/p) [1 - (-1)^{k+l}]
\nonumber \\
\geq &
\ \frac{2 |c_{00}|^2 }{N^2} \left[ 1 - \cos (2 \pi m /p) \right]
\\ & \quad -
\frac{8}{N^2} \sum_{k+l > 0 {\rm\ even}} |c_{kl}| |c_{00}|  \left[1 - \cos (2\pi m/p)\right]
-
\frac{8}{N^2} \sum_{k+l {\rm\ odd}} |c_{kl}| |c_{00}|  |\sin (2\pi m/p)|
\nonumber \\
= &
\ \frac{2|c_{00}| }{N^2} \left[ 1 - \cos (2 \pi m /p) \right] \left( |c_{00}| - 4 \sum_{k+l > 0 {\rm\ even}} |c_{kl}|
-
4 |\cot (\pi m /p)| \sum_{k+l {\rm\ odd}} |c_{kl}|  \right).
\end{align}
For $0 < m/p \leq 1/2$, $|\cot (\pi m /p)| \leq p/(m \pi)$, so
\begin{equation}
\bra{\psi} H_B \ket{\psi} \geq
\frac{2 |c_{00}|}{N^2} \left[ 1 - \cos (2 \pi m /p) \right] \left[ |c_{00}| - 4 \sum_{k+l > 0 {\rm\ even}} |c_{kl}|
- 4 \frac{p}{m\pi} \sum_{k+l {\rm\ odd}} |c_{kl}|  \right].
\label{eq:hbbound}
\end{equation}
For $1/2 \leq m/p < 1$, we can just consider $m' = p-m$, in which case $|\cot (\pi m /p)| = |\cot (\pi m'/p)|$ and $\cos (2 \pi m / p) = \cos (2 \pi m'/p)$, and we can just substitute $m'$ for $m$ in equation~(\ref{eq:hbbound}).  Otherwise, we let $m' = m$.

Therefore, one of two things must be true.  The first possibility is that the sum is dominated by $c_{00}$:
\begin{equation}
|c_{00}| - 4 \sum_{k+l > 0 {\rm\ even}} |c_{kl}| - 4 \frac{p}{m'\pi} \sum_{k+l {\rm\ odd}} |c_{kl}| \geq |c_{00}|/2,
\end{equation}
in which case
\begin{equation}
\bra{\psi} H_B \ket{\psi} \geq \frac{|c_{00}|^2}{N^2} \left[ 1 - \cos (2 \pi m' /p) \right] = \Omega\left(\frac{m'^2 |c_{00}|^2}{p^2 N^2}\right).
\end{equation}
The second possibility is that there is a non-negligible contribution from the higher $k,l$ eigenvalues:
\begin{equation}
|c_{00}| - 4 \sum_{k+l > 0 {\rm\ even}} |c_{kl}| - 4 \frac{p}{m'\pi} \sum_{k+l {\rm\ odd}} |c_{kl}| < |c_{00}|/2,
\end{equation}
which implies that
\begin{equation}
\sum_{k,l \neq (0,0)} |c_{kl}| = \Omega(m' |c_{00}|/p).
\label{eq:csum}
\end{equation}

However,
\begin{equation}
\bra{\psi} H_L + H_R \ket{\psi} = \sum_{k,l = 0}^{N-1} (\lambda_k + \lambda_l) |c_{kl}|^2.
\end{equation}
Note that
\begin{equation}
\lambda_k = 2(1 - \cos (\frac{\pi k}{N})) \geq \frac{4 k^2}{N^2}
\end{equation}
for $0 \leq k \leq N$, so
\begin{equation}
\bra{\psi} H_L + H_R \ket{\psi} \geq \frac{4}{N^2} \sum_{k,l \neq (0,0)} (k^2 + l^2) |c_{kl}|^2.
\label{eq:gridenergy}
\end{equation}
Under the constraint that $\sum_{k,l \neq (0,0)} |c_{kl}|$ is fixed (as suggested by (\ref{eq:csum})), the RHS of (\ref{eq:gridenergy}) is minimized when $|c_{kl}| = A/(k^2 + l^2)$ (this minimization can be done using Lagrange multipliers, for instance), with
\begin{equation}
\sum_{k,l \neq (0,0)} \frac{A}{k^2 + l^2} = \Omega(m' |c_{00}|/p).
\end{equation}
For large $N$, $\sum 1/(k^2 + l^2) = \Theta(\log N)$, so $A = \Omega(m' |c_{00}|/p \log N)$.  Thus,
\begin{equation}
\bra{\psi} H_L + H_R \ket{\psi} \geq \Omega(m'^2 |c_{00}|^2/p^2 N^2 \log^2 N) \sum_{k,l \neq (0,0)} \frac{1}{k^2 + l^2} = \Omega(m'^2 |c_{00}|^2/p^2 N^2 \log N).
\end{equation}
That is, regardless of the size of $|c_{00}|$, we have that $\bra{\psi} H_L + H_R + H_B \ket{\psi} = \Omega(m'^2 |c_{00}|^2/p^2 N^2 \log N)$, proving the lemma.

\end{proof}

We next turn to proving Theorem~\ref{thm:onexgap}.  The proof is somewhat similar: The ground state $\ket{{\rm ground}}$ of $H_L + H_R + H_B$ is completely delocalized in the $\ket{a,b,r}$ basis, so in particular has substantial energy for $H_P$.  In order to cancel the bad $\ket{a,b,r}$ sites, we will need in to add a nontrivial contribution of higher $k$, $l$, and $m$ values, which will increase the energy of $H_L + H_R + H_B$.

\begin{proof}[Proof of Theorem~\ref{thm:onexgap}]

Let
\begin{equation}
\ket{\psi_{klm}} = \frac{1}{\sqrt{p}} \sum_{r=0}^{p-1} e^{2\pi i mr/p} \ket{\psi_{kl},r},
\end{equation}
and
\begin{equation}
\ket{\psi} = \sum_{k,l,m = 0}^{N-1} a_{klm} \ket{\psi_{klm}}.
\end{equation}

We start by considering $|\bra{\psi_{klm}} H_P \ket{\psi_{000}}|$.  To calculate this, we expand $\ket{\psi_{klm}}$ and $\ket{\psi_{000}}$ in the $\ket{a,b,r}$ basis, since $H_P$ is diagonal in that basis.  Then $\bra{a',b',r'} H_P \ket{a,b,r}$ is $0$ unless $a=a'$, $b=b'$, and $r=r'$.  It is also $0$ if $a=1$ or $b=1$.  In all other cases, $\bra{a,b,r} H_P \ket{a,b,r} = 0$ if $\ket{\alpha_r} = \ket{\Psi^-}$ and $1$ otherwise.  Let
\begin{equation}
R = \{r \mbox{ s.t. } \ket{\alpha_r} \neq \ket{\Psi^-}\}.
\end{equation}
That is, $R$ is the set of ``bad'' $r$, that contribute to a nonzero energy for $H_P$.

Then
\begin{align}
|\bra{\psi_{klm}} H_P \ket{\psi_{000}}| = & \frac{1}{p} \left| \sum_{r,r'} e^{-2\pi i mr'/p} \bra{\psi_{kl},r'} H_P \ket{\psi_{00},r} \right| \\
= & \frac{1}{p} \left| \sum_{r \in R} e^{-2\pi i r m/p} \bra{\psi_{kl},r} H_P \ket{\psi_{00},r} \right| \\
\leq & \frac{|R|}{p} |\bra{\psi_{kl},0} H_P \ket{\psi_{00},0}|.
\end{align}
In the last line we assume, without loss of generality, that $0 \in R$.
It is also worth noting that for $m=0$, we have equality in the last line.  Furthermore,
\begin{equation}
|\bra{\psi_{00},0} H_P \ket{\psi_{00},0}| = \frac{1}{N^2} \sum_{a',a,b',b = 1}^{N} \bra{a',b',0} H_P \ket{a,b,0} = \frac{(N-1)^2}{N^2}.
\end{equation}
When $l \neq 0$,
\begin{align}
|\bra{\psi_{0l},0} H_P \ket{\psi_{00},0}| = & \frac{\sqrt{2}}{N^2} \left|\sum_{a',a,b',b = 1}^{N} \cos\left(\frac{\pi l}{N} (b'-1/2)\right) \bra{a',b',0} H_P \ket{a,b,0} \right| \\
= & \frac{\sqrt{2}(N-1)}{N^2} \left| \sum_{b=2}^{N} \cos\left(\frac{\pi l}{N} (b-1/2)\right) \right|.
\end{align}
But note that we have an orthogonality relationship for cosines, so that
\begin{equation}
\sum_{b=1}^{N} \cos\left(\frac{\pi l}{N} (b-1/2)\right) = 0,
\label{eq:orthogonal}
\end{equation}
so
\begin{equation}
|\bra{\psi_{0l},0} H_P \ket{\psi_{00},0}| = \frac{\sqrt{2}(N-1)}{N^2} \left| \cos\left(\pi l/2N\right) \right|
\end{equation}
when $l \neq 0$.  Similarly, when $k \neq 0$,
\begin{equation}
|\bra{\psi_{k0},0} H_P \ket{\psi_{00},0}| = \frac{\sqrt{2}(N-1)}{N^2} \left| \cos\left(\pi k/2N\right) \right|.
\end{equation}

Finally, when $k, l \neq 0$,
\begin{equation}
|\bra{\psi_{kl},0} H_P \ket{\psi_{00},0}| = \frac{2}{N^2} \left| \sum_{a,b=2}^{N} \cos\left(\frac{\pi k}{N} (a-1/2)\right) \cos\left(\frac{\pi l}{N} (b-1/2)\right) \right|.
\end{equation}
We use another orthogonality relationship:
\begin{align}
0 = & \sum_{a,b=1}^{N} \cos\left(\frac{\pi k}{N} (a-1/2)\right) \cos\left(\frac{\pi l}{N} (b-1/2)\right) \\
=  & \sum_{a,b=2}^{N} \cos\left(\frac{\pi k}{N} (a-1/2)\right) \cos\left(\frac{\pi l}{N} (b-1/2)\right)
+ \sum_{a=2}^{N} \cos\left(\frac{\pi k}{N} (a-1/2)\right) \cos\left(\frac{\pi l}{2N}\right) \nonumber \\
& \quad
+ \sum_{b=2}^{N} \cos\left(\frac{\pi l}{N} (b-1/2)\right) \cos\left(\frac{\pi k}{2N}\right)
+ \cos\left(\frac{\pi k}{2N}\right) \cos\left(\frac{\pi l}{2N}\right) \\
= & \sum_{a,b=2}^{N} \cos\left(\frac{\pi k}{N} (a-1/2)\right) \cos\left(\frac{\pi l}{N} (b-1/2)\right)
- \cos\left(\frac{\pi k}{2N}\right) \cos\left(\frac{\pi l}{2N}\right),
\end{align}
where in the last line, we have used the orthogonality relationship (\ref{eq:orthogonal}) for each of the two single sums.
We then get
\begin{equation}
|\bra{\psi_{kl},0} H_P \ket{\psi_{00},0}| = \frac{2}{N^2} \left| \cos\left(\frac{\pi k}{2N}\right) \cos\left(\frac{\pi l}{2N}\right) \right|.
\end{equation}

We now turn to considering $\bra{\psi} H_P \ket{\psi}$.  Let $\ket{\psi'} = \sum_{(k,l,m) \neq (0,0,0)} a_{klm} \ket{\psi_{klm}}$, so $\ket{\psi} = a_{000} \ket{\psi_{000}} + \ket{\psi'}$.  Since $H_P$ is positive definite,
\begin{align}
|\bra{\psi} H_P \ket{\psi}| \geq & |a_{000}|^2 \bra{\psi_{000}} H_P \ket{\psi_{000}} + 2\, {\rm Re} (a_{000} \bra{\psi'} H_P \ket{\psi_{000}}) \\
= & |a_{000}|^2 \frac{|R| (N-1)^2}{pN^2} + 2\, {\rm Re} \sum_{(k,l,m) \neq (0,0,0)} a_{000}\, a^*_{klm} \bra{\psi_{klm}} H_P \ket{\psi_{000}} \\
\geq & |a_{000}|^2 \frac{|R| (N-1)^2}{pN^2} - 2 \sum_{(k,l,m) \neq (0,0,0)} |a_{000}| |a_{klm}| |\bra{\psi_{klm}} H_P \ket{\psi_{000}}| \\
\geq & |a_{000}| \frac{|R| (N-1)^2}{pN^2} \left[|a_{000}| - 2 \sum_{m>0} |a_{00m}| - \frac{2 \sqrt{2}}{N-1} \sum_{k>0, m \geq 0} |a_{k0m}| \right. \nonumber \\
& \qquad \left. - \frac{2 \sqrt{2}}{N-1} \sum_{l>0, m \geq 0} |a_{0lm}| - \frac{4}{(N-1)^2} \sum_{k,l > 0, m \geq 0} |a_{klm}| \right] \\
= & |a_{000}| \frac{|R| (N-1)^2}{pN^2} \left[|a_{000}| - X - Y - Z \right],
\end{align}
where
\begin{align}
X = & \frac{2 \sqrt{2}}{N-1} \sum_{k>0} |a_{k00}| + \frac{2 \sqrt{2}}{N-1} \sum_{l>0} |a_{0l0}| + \frac{4}{(N-1)^2} \sum_{k,l > 0} |a_{kl0}| \\
Y = & \sum_{m>0} \left[\frac{2 \sqrt{2}}{N-1} \sum_{k>0} |a_{k0m}| + \frac{2 \sqrt{2}}{N-1} \sum_{l>0} |a_{0lm}| + \frac{4}{(N-1)^2} \sum_{k,l > 0} |a_{klm}| \right]\\
Z = & 2 \sum_{m>0} |a_{00m}|.
\end{align}

There are four cases:

\begin{itemize}

\item {\bf If $\mathbf{|a_{000}| - X - Y - Z \geq 1/2}$:} Then $|a_{000}| > 1/2$ as well, so $|\bra{\psi} H_P \ket{\psi}| \geq \frac{|R| (N-1)^2}{4pN^2} = \Omega(1/p)$.

\item {\bf If $\mathbf{X = \Theta(1)}$:} At least one of the sums $\frac{1}{N-1} \sum_{k>0} |a_{k00}|$, $\frac{1}{N-1} \sum_{l>0} |a_{0l0}|$, and $\frac{1}{(N-1)^2} \sum_{k,l > 0} |a_{kl0}|$ must also be $\Theta(1)$; the others may be smaller.  These three sums are all averages, and at least one of the terms in an average must be greater than or equal to the mean value.  Thus, there exists $(k_0,l_0) \neq (0,0)$ with $|a_{k_0 l_0 0}| = \Theta(1)$.  But
    \begin{equation}
    \bra{\psi} H_L + H_R \ket{\psi} = \sum_{k,l,m} |a_{klm}|^2 \lambda_{kl} \geq |a_{k_0 l_0 0}|^2 \Omega(1/N^2) = \Omega(1/N^2).
    \end{equation}

\item {\bf If $\mathbf{Y = \Theta(1)}$:}
    At least one of the sums $\frac{1}{N-1} \sum_{m>0} \sum_{k>0} |a_{k0m}|$, $\frac{1}{N-1} \sum_{m>0} \sum_{l>0} |a_{0lm}|$, and $\frac{1}{(N-1)^2} \sum_{m>0} \sum_{k,l > 0} |a_{klm}|$ must be $\Theta(1)$.  We have
    \begin{equation}
    \bra{\psi} H_L + H_R \ket{\psi} \geq \sum_{(k,l) \neq (0,0)} \sum_m |a_{klm}|^2 \Omega(1/N^2).
    \end{equation}
    When $\sum |a_{klm}|$ is fixed (summed over some subset of $(k,l,m)$), we minimize the sum $\sum |a_{klm}|^2$ (over the same subset of $(k,l,m)$) by taking all $|a_{klm}|$ to be equal (for $(k,l,m)$ in the subset).  In this particular case, we should choose $|a_{klm}| = \Theta(1/(p-1))$ for one of the sets $\{k>0, l=0, m>0\}$, $\{k=0, l>0, m>0\}$, or $\{k>0, l>0, m>0\}$ and we find $\sum_{klm} |a_{klm}|^2 = \Omega(N/p)$.  (Or $\Omega(N^2/p)$ in the third case.)  In any case, it certainly follows that $\bra{\psi} H_L + H_R \ket{\psi} = \Omega(1/N^2)$.

\item {\bf If $\mathbf{Z = \Theta(1)}$:}
    This is the most interesting case.  We note that
    \begin{equation}
    \bra{\psi} H_L + H_R + H_B \ket{\psi} = \sum_m \bra{\psi_m} H_L + H_R + H_B \ket{\psi_m},
    \end{equation}
    with $\ket{\psi_m} = \sum_{k,l} a_{klm} \ket{\psi_{klm}}$.
    By Lemma~\ref{lemma:highm},
    \begin{equation}
    \bra{\psi_m} H_L + H_R + H_B \ket{\psi_m} = \Omega(m'^2 |a_{00m}|^2/p^2 N^2 \log N),
    \end{equation}
    where $m' = \min(m, p-m)$.  Therefore, we wish to minimize $\sum_m m'^2 |a_{00m}|^2$ under the constraint that $\sum_{m>0} |a_{00m}| = \Theta(1)$.  The minimum is achieved (as can be shown by Lagrange multipliers, for instance) when $|a_{00m}| = B/m'^2$.  Since $\sum_m 1/m'^2 \leq 2 \zeta(2) = \pi^2/3$, it follows that $B = \Theta(1)$ and $\sum_m m'^2 |a_{00m}|^2 = \Omega(1)$.  Thus,
    \begin{equation}
    \bra{\psi} H_L + H_R + H_B \ket{\psi} = \Omega(1/p^2 N^2 \log N).
    \end{equation}

\end{itemize}

As it happens, case 2 cannot actually occur in the large $N$ limit, given the normalization constraint that $\sum |a_{klm}|^2 = 1$, but this is not very important since the gap is set by the fourth case.

\end{proof}

\subsection{Zero or Multiple $\ket{x}$ States}

We now turn attention to the subspaces $\Hil_{a,b}$, with either $a$ or $b$ different than $1$.  We can treat $H_V^L$ and $H_V^R$ separately, and show that $H_L + H_V^L$ is bounded below when $a \neq 1$, as is $H_R + H_V^R$ when $b \neq 1$.  Since the remaining Hamiltonian terms $H_B$ and $H_P$ are non-negative, bounding $H_L + H_V^L$ and $H_R + H_V^R$ is adequate to give a lower bound on $H$.  Since the two cases are identical, we focus on the left ring.

If $a=0$, $H_V^L$ has all eigenvalues equal to $V_1$ and $H_L$ is $0$.  If $a > 1$, there are some states in $\Hil_{a,b}$ which have a negative eigenvalue for $H_V^L$. However, those states will have a positive eigenvalue for $H_L$, and we must choose $V_1$ and $V_2$ to give such states a positive overall energy.

Let us restrict attention to a particular value of $a$; then $H_V^L + (a-1)V_1$ is a non-negative operator, which allows us to apply the following lemma (see \cite{Kitaevbook} for a proof):
\begin{lemma}
\label{lemma:Kitaev}
Let $A_1$ and $A_2$ be nonnegative linear operators with null spaces $L_1$ and $L_2$.  Suppose that $L_1 \cap L_2 = \{0\}$ and that no nonzero
eigenvalue of $A_1$ or $A_2$ is less than $v$.  Then
\begin{equation}
A_1 + A_2 \geq v \sin^2 (\theta/2),
\end{equation}
where $\theta = \theta(L_1, L_2)$ is the angle between $L_1$ and $L_2$.
\end{lemma}
For us, $A_1 = H_V^L + (a-1)V_1$ and $A_2 = H_L$.  The null space $L_1$ of $A_1$ consists of superpositions of states with $a$ $\ket{x}$ states in the left ring with no two of those $\ket{x}$ states are adjacent.  The smallest nonzero eigenvalue of $A_1$ is at least $V_2$. (It could potentially be higher than $V_2$, if $a$ is large enough so that all configurations contain multiple adjacent $\ket{xx}$ pairs in the left ring.)  The null space $L_2$ of $A_2$ consists of uniform superpositions over all possible configurations of $a$ $\ket{x}$ sites on the left ring, with some state for the remaining qubit sites, shuffled around appropriately for different $\ket{x}$ configurations.

To find the lowest nonzero eigenvalue of $A_2$, we can recognize $H_L$ as the spin-$1/2$ ferromagnetic Heisenberg model on a chain of length $N$ with open boundary conditions by mapping $\ket{x}$ to spin up and a qubit state to spin down.  In this case, we now consider the $\ket{x}$ states to represent particles hopping around on the chain.  Unlike the usual ferromagnetic Heisenberg model, the background states with no particles have some structure, since the background states are actually qubits.  However, the qubits are inert and cannot hop or otherwise interact by themselves; they can only switch places with a $\ket{x}$ particle as it hops.  Therefore, the Hilbert space breaks up into invariant subspaces within which the state of the background sites is determined completely by the locations of the $\ket{x}$ particle states.  Within a single such subspace, we have exactly the usual spin-$1/2$ ferromagnetic Heisenberg model, which is known to have gap $1 - \cos(\pi/N) = \Theta(1/N^2)$~\cite{Heisenberg}.  (This gap considers all possible numbers of $\ket{x}$ states; in fact, the minimum gap is achieved when there is just one $\ket{x}$.)

The angle $\theta$ between $L_1$ and $L_2$ can be determined by simply noting that there are $\binom{N}{a}$ total possible $\ket{x}$ configurations, but at most \begin{equation}
\frac{N(N-2)(N-4)(N-6) \cdots (N-2a+2)}{a!} \leq \binom{N}{a} \left(1 - \frac{1}{N-1} \right)
\end{equation}
configurations with no adjacent $\ket{xx}$ pairs.  This formula comes from choosing the location of $a$ $\ket{x}$ sites, with each one chosen from possibilities excluding the sites of previously chosen $\ket{x}$ sites as well as the next higher-numbered site to any previously chosen site; then we divide by $a!$ to eliminate the redundancy due to choosing the sites in order.  This of course is an overestimate, as we must also exclude the next lower-numbered site to any previously chosen site, but then the counting becomes complicated, as the site below one previously chosen site can be the same as the site above a previously chosen site.  Since an element of $L_2$ consists of a superposition over all $\binom{N}{a}$ configurations and an element of $L_1$ can only contain configurations with no adjacent pairs, it follows that $\sin^2 \theta \geq \frac{1}{N-1}$.

Putting these numbers into lemma~\ref{lemma:Kitaev}, we find that on the subspace with $a>1$, $H_L + H_V^L + (a-1)V_1$ has all eigenvalues $\Omega(1/N^3)$.  If we choose $V_1 = \Theta(1/N^4)$ and $V_2 = 1$, then $H_L + H_V^L$ also has all eigenvalues $\Omega(1/N^3)$ on the subspace with $a>1$, while on the sector with $a=0$, $H_L$ will have eigenvalue $\Theta(1/N^4)$.  Putting everything together, we find that the ground state is unique, and is the desired ground state $\ket{g}$ (with energy $0$), while the gap in the whole Hilbert space is $\Omega(1/N^4 \log N)$.

\subsection{Entanglement of the Ground State}

The ground state of $H$ is thus
\begin{equation}
\ket{g} = \frac{1}{N} \sum_{a,b = 1}^{N} M_{a,b} \ket{xx}_{N}^{LR} \bigotimes_{i=1}^{N-1} \ket{\Psi^-}_i^{LR},
\end{equation}
as noted before.  Suppose we make a local measurement to each ring that measures the location of the $\ket{x}$ state but nothing else about the qubit states.  The outcomes are uniformly distributed among all possible pairs $(a,b)$, and given a specific outcome, the residual state is
\begin{equation}
M_{a,b} \ket{xx}_{N}^{LR} \bigotimes_{i=1}^{N-1} \ket{\Psi^-}_i^{LR}.
\end{equation}
$M_{a,b}$ is just a tensor product of a local unitary operator $M_a$ on the left ring with a local unitary operator $M_b$ on the right ring. By performing $M_a^\dagger$ and $M_b^\dagger$, we can therefore convert the state to $\ket{\Psi^-}^{\otimes (N-1)}$.  Since purely local operations convert $\ket{g}$ into $N-1$ EPR pairs, the original state has at least $N-1$ ebits of entanglement.  That is, $S(\rho_L) = S(\rho_R) \geq N-1$.  Since one ring has only $N$ $3$-state particles, the maximum possible entropy would be $N \log_2 3$, so we find that $S(\mbox{one ring}) = \Theta(N)$.  If we rephrase this in terms of the energy gap, $\Delta = \Omega(1/N^4 \log N)$, we see that $S(\mbox{one ring}) = \Omega((-\Delta \log \Delta)^{-1/4})$, as claimed.

\section{Discussion}

The explicit construction in this paper shows that it is possible to have the entanglement entropy scale polynomially with the inverse gap.  On the other hand, it was proven previously that the entanglement entropy cannot scale faster than exponentially in the inverse gap.  A tight bound remains open.

One route to a tight bound might be to improve the Lieb-Robinson bound.
This bound is a key ingredient in proving upper bounds on entanglement entropy, but our construction suggests that the Lieb-Robinson bound does not fully capture the locality properties of Hamiltonian dynamics.
Let us explain by an example: consider a system of free fermions on an infinite line in one dimension, with Hamiltonian $\sum_i \Psi_{i+1}^{\dagger} \Psi_i + h.c.$ Consider the operator equations of motion for the operator $\Psi^{\dagger}_0$.  The Lieb-Robinson bound says that $\exp[i H t] \Psi^{\dagger}_0 \exp[-i H t]$ can be well-approximated by an operator supported just on the set of sites within distance $v_{LR} t$ of $0$, where $v_{LR}$ is the Lieb-Robinson velocity.  However, at least for this system, we can make a sharper statement --- there is no good approximation by an operator supported on too {\em few} sites.  The exact
result for this operator is:
\begin{equation}
\exp[i H t]\Psi^{\dagger}_0 \exp[-i Ht]=\sum_x A(x,t) \Psi^{\dagger}_x,
\end{equation}
where
\begin{equation}
A(x,t)\equiv \int \frac{{\rm d}k}{2\pi} \exp[2i\cos(k t)] \exp(i kx).
\end{equation}
One can check from the asymptotics of the Bessel function that for $x>>2t$ the
amplitude $A(x,t)$ is exponentially small, in agreement with the Lieb-Robinson bound.  However, the
amplitude $A(x,t)$ is symmetric under a change in sign from $x \rightarrow -x$.
That is,
the operator $ \Psi^{\dagger}_0$ creates both left and right moving excitations.
Further, while the amplitude is maximum at $x$ close to $\pm 2t$, there is a
non-vanishing spread in the width of $|A(x,t)|^2$: the particle is not perfectly
localized at $x=\pm 2t$, but instead has a spread in its position, indicating
an uncertainty in the initial velocity of the particle.
  It is possible to create a wavepacket that moves to the right, but in order to do this we need an operator that creates the particle in a superposition of different
sites: $\sum_i a(i) \Psi^{\dagger}_i$ for some function $a(i)$.  If $a(i)$ has finite support, then there will always be some spread in the velocity of the wavepacket, and in fact there will be some left-moving component to the wavepacket.
Thus, it will not be possible to find, at least in this model, a choice of
$a(i)$ which has finite support for which the time evolved operator
$\exp[i H t] \Bigl(\sum_i a(i) \Psi^{\dagger}_i \Bigr) \exp[-i H t]$ is approximately localized on a set of sites
with a bounded, time-independent diameter.
The Lieb-Robinson bound does not capture this effect, but it seems to be a real effect for every Hamiltonian that we can think of: operators with finite support create both left- and right-moving excitations.

A similar phenomenon afflicts our construction.  We had to go to some effort to keep the two rings synchronized because of the spread in the velocity of the $|x\rangle$, which, if left unchecked, could wash out the entanglement between the two rings.
To do this, we only allowed
the $|x\rangle$s to hop from site $1$ to site $N$ in pairs.
The fact that the term $H_B$ only
has effect when both $|x\rangle$s are in the correct position is one of the reasons that the gap is so small.
However, if we had a way of just creating an $|x\rangle$ excitation that moved to the right around the ring at a constant velocity equal to $v_{LR}$, it might be easier to synchronize the rings, and thus to increase the gap for the same entanglement.  Hopefully, then, this work will lead to a better Lieb-Robinson bound that captures these effects, and this in turn will lead to stronger bounds on entanglement in Hamiltonian systems.

\subsection*{Acknowledgements}  We thank D. Aharonov for many useful
discussions throughout this work.
DG was supported by CIFAR, by the Government of Canada through NSERC, and by the Province of Ontario
through MRI.  MBH was supported by U.~S.\ DOE Contract No. DE-AC52-06NA25396.

\end{document}